\documentclass[11pt,leqno]{article}
\usepackage[letterpaper,margin=1in]{geometry}
\usepackage{amsmath,amssymb,amsfonts,amscd,amsthm,upref}
\usepackage{graphicx}
\usepackage{tikz}
\usetikzlibrary{trees,positioning,quotes}
\usepackage{comment}
\usepackage[shortlabels]{enumitem}
\usepackage{algorithm}
\usepackage{algpseudocodex}
\usepackage[hidelinks]{hyperref}
\usepackage[capitalize,nameinlink]{cleveref}
\newtheorem{theorem}{Theorem}[section]
\newtheorem{lemma}[theorem]{Lemma}

\newtheorem{definition}[theorem]{Definition}
\newtheorem{observation}[theorem]{Observation}

\theoremstyle{remark}

\numberwithin{equation}{section}
\numberwithin{table}{section}
\numberwithin{figure}{section}
\numberwithin{algorithm}{section}
\crefname{section}{section}{sections}
\crefname{subsection}{subsection}{subsections}
\Crefname{section}{Section}{Sections}
\Crefname{subsection}{Subsection}{Subsections}
\crefname{observation}{Observation}{Observations}
\Crefname{observation}{Observation}{Observations}
\Crefname{figure}{Figure}{Figures}

\newcommand{\cut}{\text{Cut}}
\DeclareMathOperator*{\argmax}{arg\,max}
\DeclareMathOperator*{\argmin}{arg\,min}

\allowdisplaybreaks
\newcommand{\arxivpublicationnote}{%
  \begingroup
  \renewcommand{\thefootnote}{}%
  \footnotetext{To appear in SODA 2027. Copyright © 2027 by SIAM. Unauthorized reproduction of this article is prohibited.}%
  \endgroup
}
\date{}
\hypersetup{pdftitle={No, Cake Cutting Really is a Piece of Cake},
  pdfauthor={Stephen Arndt, Sungjin Im, Benjamin Moseley, Kirk Pruhs}}

\title{No, Cake Cutting  Really is a Piece of Cake}

\author{Stephen Arndt\thanks{Carnegie Mellon University, Tepper School of Business. Supported in part by NSF grant CCF-2550399.} \and Sungjin Im\thanks{University of California, Santa Cruz. Supported in part by NSF grants CCF-2423106, CCF-2537126, and CCF-2535599.} \and Benjamin Moseley\thanks{Carnegie Mellon University, Tepper School of Business. Supported in part by NSF grant CCF-2550399. } \and Kirk Pruhs\thanks{University of Pittsburgh. Supported in part by NSF grants CCF-2550399 and CCF-2209654.} }

\begin{document}
\maketitle
\arxivpublicationnote
\begin{abstract}
   We design and analyze a deterministic cake cutting algorithm that achieves proportional fairness using a linear number of cuts. The best previous  upper bound on the number of cuts for a deterministic algorithm was $O(n \log n)$, 
   which was obtained by a natural divide-and-conquer algorithm  due to Even and Paz. It has long been conjectured  
   that $O(n \log n)$ cuts was optimal for a deterministic algorithm. 
\end{abstract}

\section{Introduction}\label{sec:intro}

\subsection{Problem Statement}

Intuitively cake cutting problems involve $n$ self-interested players who want
to divide a divisible resource, such as a cake. Here we consider achieving the objective of proportional
fairness in the standard Robertson-Webb 
model~\cite{RobertsonWebbbook,JiriGerhardComplexity}. 
So formally an instance of the cake cutting problem consists of $n$ value functions
$\mu_1, \ldots, \mu_n$, where each $\mu_p$ is from the unit interval $(0, 1)$ to  the
nonnegative reals, and where $\int_{x=0}^1 \mu_p(x) dx = 1$. Intuitively 
the unit interval is the cake, and
$V_p(a, b) = \int_{x=a}^b \mu_p(x) dx $ is how much player $p$ values the portion $(a, b)$ of the cake. 
A feasible solution is a partition of the cake, and a proportionally fair assignment
of the pieces to the players, 
that is each player must be assigned a part of the cake that is of value at least $\frac{1}{n}$ to them. 
Each part in the standard algorithms is a subinterval, but this is not a requirement of the problem, and more generally a
part could be the union of subintervals. 
The algorithm does not initially know the value functions, and can only learn information
about the value functions by querying the players. The most important type of query is the
Cut query, which is of the form $\cut_p(\alpha)$, where $p$ is a player and
$\alpha \in [0, 1]$ is a value. 
The expected response from  $p$ is 
the cut point $y$ such that $V_p(0, y) = \alpha$. 
The second type of query is 
an evaluation query $E_p(y)$, which consists of a player $p$ and a point $y$ that was some player's (not necessarily player $p$) response
to a prior cut query, and 
the expected response from player $p$ is $V_p(0, y)$. 
Initially it will be convenient to assume that all players always answer truthfully, and that
for all players $p$ and positions $a < b$ it is the case that $V_p(a, b) > 0$.

\subsection{Essential Background}

The cake cutting literature is quite large, including several books and surveys dedicated to the topic~\cite{RobertsonWebbbook,Procaccia_2016,Brams_Taylor_1996}.
Here we just cover the results that are most essential for our purposes, starting with the standard algorithms for proportional fairness. 

In the 1940s Banach and Knaster designed a deterministic algorithm, sometimes called the
Last Diminisher algorithm, that
uses $O(n^2)$ cuts, and this algorithm  was communicated by Steinhaus
in 1948~\cite{steinhaus}. This algorithm uses $n$ cuts to find the player $p$ that most values the left portion of the cake,
more precisely $p = \argmin_q \cut_q( \frac{1}{n})$.
Player $p$ is then assigned the piece $\left(0, \cut_p\left(\frac{1}{n}\right)\right)$ of the cake.
The algorithm then recurses on the rest of the players and the remaining unassigned portion of the cake. 
Steinhaus~\cite{steinhaus} noted that:
\begin{quote}
``Interesting
mathematical problems arise if we are to determine the minimal numbers of `cuts' necessary for fair division.''
\end{quote}
In 1984 Even and Paz designed a deterministic divide-and-conquer algorithm 
that
uses $O(n \log n)$ cuts~\cite{EvenPaz}.
Even and Paz's algorithm uses $n$ cuts to determine a split position $y$, and the  $n/2$ players $L$ that will appear
in the first half in the left-to-right order of the pieces, and the  $n/2$ players $R$ that will appear
in the second half of  the left-to-right order of the pieces. The algorithm then recurses to determine how
the players in $L$ divide the portion $(0, y)$ of the cake and recurses to determine how the players in 
$R$ divide the portion $(y, 1)$ of the cake. 
In this same paper~\cite{EvenPaz} Even and Paz also designed a randomized
 divide-and-conquer algorithm that uses $O( n)$ cuts in expectation.
The algorithm uses $O(\log n)$ cuts  in expectation
to determine the split position $y$, and then proceeds in the same manner as the deterministic algorithm.
The key to the correctness proofs for the Last Diminisher algorithm, and both of Even and Paz's algorithms, is that all recursive subproblems
have the following abundance condition: 
\begin{description}
    \item[Abundance condition:] If a subproblem involves $k$ players dividing a portion of the cake, 
then each player values that portion at least $\frac{k}{n}$.
\end{description}

The most straightforward implementation of  the Last Diminisher algorithm can be implemented in the Robertson-Webb model using only cut operations (so no evaluation operations).
Even and Paz~\cite{EvenPaz} didn't bother to evaluate the  number of evaluations in their algorithms,
as the general sentiment of researchers working on cake cutting was that the primary objective
was to limit the number of cuts. In the chapter entitled ``Pieces or Crumbs --- How Many Cuts are Needed?'', of their  1998 textbook ``Cake-Cutting Algorithms: Be Fair if you Can''~\cite{RobertsonWebbbook},  Robertson and Webb 
stated the case for this sentiment as follows:
\begin{quote}
    ``Certainly if you are actually dividing a cake, you don't want to receive a pile of crumbs for your fair share.
    You prefer not to make too many cuts. More realistically, the issue at hand is a mathematical one.''
\end{quote}
The  most straightforward implementation of Even and Paz's deterministic algorithm would use 
$O(n \log n)$ evaluations, but it is also implementable with no evaluations. 
A straightforward implementation of Even and Paz's randomized algorithm would use 
$O(n \log n)$ evaluations in expectation.

Subsequently  the central open question in this area became whether there was a deterministic algorithm 
that uses $O(n)$ cuts (and an arbitrary number of evaluation queries). In 1984 Even and Paz~\cite{EvenPaz} stated:
\begin{quote}
``We conjecture that no fair deterministic protocol exists in which the number of cuts is $O(n)$.''
\end{quote}
In their 1998 textbook~\cite{RobertsonWebbbook}  Robertson and Webb 
doubled down with an even stronger  conjecture,
saying that
\begin{quote}
   ``We have lost bets before, but if we were asked to gaze into the crystal ball, we would place our money against finding a substantial improvement on the 
   $ n \log_2 n$  bound.''
\end{quote}

Efforts to prove the non-existence of an $O(n)$ cut deterministic 
algorithm began in the 1990's, with the necessary first
step being defining a formal model, which wasn't strictly necessary for the
prior upper bound results in the literature. 
An initial attempt at defining a formal model was made in the 1990's by Robertson and 
Webb~\cite{RobertsonWebb1995,RobertsonWebbbook}, and this was later refined in the 2000's by Sgall and Woeginger~\cite{JiriGerhardComplexity}.
This model was called the Robertson-Webb model in \cite{JiriGerhardComplexity}. 
With the Robertson-Webb model established, 
Sgall and Woeginger~\cite{JiriGerhardComplexity} gave an $\Omega(n \log n)$ lower bound on the number
of cut queries plus the number of evaluation queries for any deterministic protocol 
that assigns the players  pieces that are subintervals of the cake. 
Soon after, Edmonds and Pruhs~\cite{EdmondsPruhs} extended Sgall and Woeginger's $\Omega(n \log n)$ lower bound  to apply to deterministic protocols in which the parts may be a collection of subintervals. 
If one takes the view that cut queries and evaluation queries are roughly 
equally valuable/limiting resources, then this closes the story for deterministic algorithms,
as $\Theta(n \log n)$ queries are necessary and sufficient to achieve proportional fairness
in the Robertson-Webb model. If instead one takes the view of Even and Paz~\cite{EvenPaz}, and Robertson and Webb~\cite{RobertsonWebbbook}, that the primary objective should be to minimize the number of
cut queries, then the lower bounds of Sgall and Woeginger~\cite{JiriGerhardComplexity}, and Edmonds and Pruhs~\cite{EdmondsPruhs},   don't directly address whether there
is an $O(n)$ cut deterministic algorithm.

\subsection{Our Results}

Our main contribution is to refute the conjecture of Even and Paz~\cite{EvenPaz}, and collect on the bet of  Robertson and Webb~\cite{RobertsonWebbbook}, by giving a deterministic
algorithm that achieves proportional fairness using $O(n)$ cuts. 

\begin{theorem}\label{thm:main}
    There is a deterministic algorithm in the Robertson-Webb cake cutting model, which we call the Late-Early algorithm, that achieves proportional
    fairness using $O(n)$ cuts
    and $O(n^2)$ evaluations. 
\end{theorem}

In Section \ref{sec:story}, we tell the story behind the design and analysis of the Late-Early algorithm, as
we believe that it is interesting. This section also serves as an overview for readers seeking a high-level overview before 
diving into the details of the proof of 
 Theorem \ref{thm:main}, which  we give in Section \ref{section:earlylate}.
To give some more context to    Theorem \ref{thm:main}, in Section \ref{sec:nonessential} we highlight a few other results in 
the voluminous literature on cake cutting and fair division.

\section{The Story of the Algorithm Development }
\label{sec:story}

As intuitively the hardest instances of finding one of $k$ needles in a haystack are those with $k=1$, it is natural to conjecture that the hardest instances of cake cutting are those with (essentially) a unique solution. This intuition underlies the $\Omega(n\log n)$ lower bound of \cite{JiriGerhardComplexity} for cuts plus evaluations. This lower bound considers a collection $\mathcal I=\{I_\pi\}$ of instances, one for each permutation $\pi:[n]\to[n]$. Each instance $I_\pi$ has (essentially) a unique solution $S_\pi$ in which every player receives a part that is a subinterval, and the left-to-right order of subintervals is given by the permutation $\pi$. Specifically, player $p$ is assigned the $\pi(p)$'th subinterval on the cake. This lower bound gives an adversarial strategy showing that, for every deterministic algorithm $A$, there is an instance $I_\pi$ on which either $A$ is incorrect on $I_\pi$, or $A$ uses $\Omega(n\log n)$ queries.

Our original goal was to prove Even and Paz's conjecture that there is no deterministic algorithm that uses $O(n)$ cuts. We therefore attempted to prove such a lower bound for algorithms that assign subintervals using the instances $\mathcal I$. Since \cite{JiriGerhardComplexity} show that one may assume all cut queries are at integer multiples of $\frac1n$, our first observation was an $\Omega(n\log n)$ lower bound for position-oriented deterministic algorithms. Intuitively, an algorithm $A$ is position-oriented if, for every $j$, it is the case that all cut queries of value $\frac{j}{n}$ occur consecutively in the execution of $A$, so the algorithm $A$ completely explores one position for the players before moving to another. The algorithms of Steinhaus~\cite{steinhaus} and Even and Paz~\cite{EvenPaz} are position-oriented. We suspect this observation may already have been known, and may even in part have motivated the conjecture that there is no $O(n)$ cut deterministic algorithm. However, we found that the argument did not extend to more player-oriented algorithms, which intuitively are algorithms whose cuts can  be partitioned into consecutive blocks involving a single player (instead of blocks involving a single position, as in the position-oriented algorithms).

Thus, guided by our failure to extend our lower bound for position-oriented algorithms to player-oriented algorithms,
we switched roles, and attempted to find a deterministic  algorithm that was more player-oriented. 
So such an algorithm would fix a player  $p$, and then ask player $p$ a sequence of cut
queries until some useful information is gained from $p$. 
A piece of information that would be sufficient is identifying a ``late-early'' piece with no inside players. Formally, define $y_{p,j}$ as the point on the cake such that the value to the right of this point is $\frac{n-j}{n}$ according to player $p$. 
A late-early piece is  a subinterval of the form $\left(y_{p, j}, y_{p, j+1}\right)$, where $y_{p, j}$ is ``late'' in the multiset $\{y_{q,j} \}$ (meaning $y_{p, j}$ is not one of the smallest $j$ points in $\{y_{q,j} \}$) 
and $y_{p, j+1}$ is ``early'' in the multiset $\{y_{q,j+1} \}$ (meaning $y_{p, j+1}$ is one of the smallest $j+1$ points in $\{y_{q,j+1}\}$). 
Further, a late-early piece  $\left(y_{p, j}, y_{p, j+1}\right)$ can be found in $O(\log \min(j, n-j))$ cuts using a standard doubling search.
An inside player for a piece $(y_{p, j}, y_{p, j+1})$ is a player $q \neq p$ such that the interval $\left(y_{q, j}, y_{q, j+1}\right) $ is contained in the interval $ \left(y_{p, j}, y_{p, j+1}\right)$. See Figure \ref{fig:lateearlysimple} for an illustration. 

\begin{figure}[h]
\begin{center}
\begin{tikzpicture}
\filldraw[black] (0,0) circle (2pt) node[anchor=north]{$y_{p,j}$};
\draw[black, thick] (0,0) -- (6,0);
\filldraw[black] (6,0) circle (2pt) node[anchor=north]{$y_{p,j+1}$};
\filldraw[black] (-2,-1) circle (2pt) node[anchor=north]{$y_{s,j}$};
\draw[black, thick] (-2,-1) -- (4,-1);
\filldraw[black] (4,-1) circle (2pt) node[anchor=north]{$y_{s,j+1}$};
\filldraw[black] (2,-2) circle (2pt) node[anchor=north]{$y_{r,j}$};
\draw[black, thick] (2,-2) -- (8,-2);
\filldraw[black] (8,-2) circle (2pt) node[anchor=north]{$y_{r,j+1}$} ;
\filldraw[black] (2,-3) circle (2pt) node[anchor=north]{$y_{q,j}$};
\draw[black, thick] (2,-3) -- (4,-3);
\filldraw[black] (4,-3) circle (2pt) node[anchor=north]{$y_{q,j+1}$};
\filldraw[black] (-2,-4) circle (2pt) node[anchor=north]{$y_{t,j}$};
\draw[black, thick] (-2,-4) -- (8,-4);
\filldraw[black] (8,-4) circle (2pt) node[anchor=north]{$y_{t,j+1}$} ;
\end{tikzpicture}
\end{center}
\caption{In this instance player $q$ is an inside player for the piece $(y_{p,j}, y_{p, j+1})$, but none of players $r$, $s$, or $t$ are inside players for the piece $(y_{p,j}, y_{p, j+1})$. In this instance, player $t$ is an outside player for the piece $(y_{p,j}, y_{p, j+1})$, but none of players $q$, $r$, or $s$ are outside players for the piece $(y_{p,j}, y_{p, j+1})$.  }
\label{fig:lateearlysimple}
\end{figure}
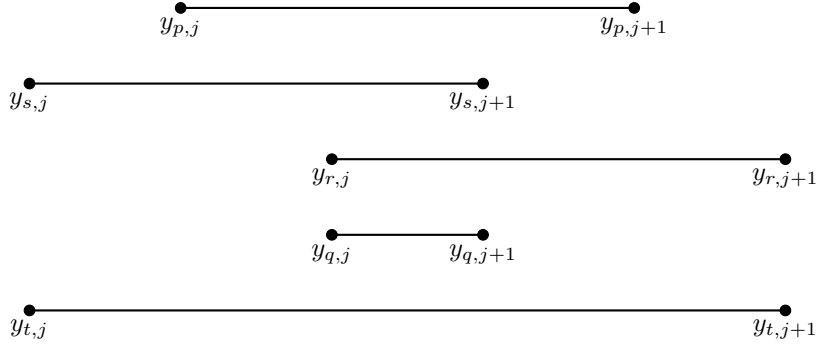

If the algorithm could identify a late-early piece $\left(y_{p, j}, y_{p, j+1}\right)$ with no inside players, then the algorithm could assign the player $p$ the late-early piece $\left(y_{p,j}, y_{p,j+1}\right)$, then recurse on 
$(0, y_{p,j})$ with  $j$  players, and recurse on $( y_{p,j+1}, 1 )$ with  $n-j-1$ players. For brevity we omit the details of how the players are partitioned, but the details are relatively simple to work out; as a warm-up, observe that for  the corner  case that 
 $y_{p, j}$ is last in  $\{y_{q,j} \}$ and $y_{p, j+1}$ is first in $\{y_{q,j+1} \}$, the $j$ players recursing left can be selected arbitrarily. 
For the lower bound instances $I_\pi$ from \cite{JiriGerhardComplexity},  a late-early piece $\left(y_{p, j}, y_{p, j+1}\right)$  without inside players must exist.
Further, it can be found in $O(\log \min(j, n-j))$ cuts using a standard doubling search.
This yields a deterministic algorithm for the instances $I_\pi$ from \cite{JiriGerhardComplexity} 
that uses $O(n)$ cuts. 
The correctness of the algorithm follows from the fact that both recursive subproblems satisfy the abundance condition.

Now this cannot directly lead to an algorithm on general instances, because it can be the case that all of the late-early pieces  contain inside players. Further one can see that in these instances there is no way, using only local information, to determine left and right subproblems that will be guaranteed to have the abundance condition.  Thus at first glance it might appear that proceeding recursively cannot work as 
the abundance condition seems critical to guaranteeing that 
the recursive subproblems admit a feasible solution. 

But after some reflection, one can observe that for every inside player $q$ for the late-early piece  $ \left(y_{p, j}, y_{p, j+1}\right)$,
there is an outside player $t$. Player $t$ is
an outside player  for a late-early piece  $ \left(y_{p, j}, y_{p, j+1}\right)$ if the
interval $\left(y_{p, j}, y_{p, j+1}\right) $ is contained in the interval $ \left(y_{t, j}, y_{t, j+1}\right)$. Intuitively outside players are agnostic about whether they want to be part of the left recursive call or the right recursive call.
Thus the existence of an inside player (essentially) implies that this
instance does not have a unique solution (because outside players can go either left or right).
Intuitively this should imply that there is more than one needle in the  haystack under consideration, and thus, intuitively 
this should make the  problem of finding a needle in this haystack easier. 
After more reflection, we eventually observed  that one can slightly strengthen the definition of a late-early piece $\left(y_{p,j}, y_{p,j+1}\right)$, in such a way that
one can still find such a piece using the standard doubling search, and for which it is possible, using only local information, 
 to proceed recursively
on left and right instances.  Although these left and right instances may not satisfy the abundance condition, one can  still 
show that they must admit feasible solutions. 
The type of certificate of feasibility for 
these instances that we will use in our algorithm analysis is what we call the subsistence condition,
which roughly  means that there is 
a player $q$  
that values one  end of the cake so much that cutting a proportionately
fair piece off that end for $q$ leaves a resulting subproblem that has the
abundance condition. 
So intuitively an instance with the subsistence condition is one proportionally fair piece assignment away from becoming abundant.~\footnote{An alternate intuitive interpretation of   the subsistence condition is 
 that  it means that Steinhaus' Last Diminisher algorithm
would produce a feasible solution on the instance. }

In summary, at a high level we initially designed an algorithm that was correct if there was a unique solution.
Then when this algorithm failed, the algorithm learned a certificate
that the instance had multiple solutions. The algorithm then  continued the search recursively on subinstances, 
that plausibly could have been infeasible before learning the certificate, but that the certificate now implies
must have a feasible solution. 
We are not aware of other algorithms in the algorithmic foundations literature that use such a design paradigm,
and it seems at least plausible that this algorithm design paradigm could be useful in other settings.

\section{The Late-Early Algorithm}
\label{section:earlylate}

In \Cref{subsec:generalizedinduction}, we define the generalized inductive
hypothesis / problem that we design an algorithm for, and  the subsistence condition of instances that
is the key to our correctness analysis. 
In \Cref{subsec:description}, we give a  description of our
Late-Early algorithm. 
In \Cref{subsec:implementation}, we explain  some
lower-level implementation details.
In \Cref{subsec:correctness}, we prove the correctness of
our Late-Early algorithm. 
In  \Cref{subsec:complexity}, we do a complexity analysis of
our Late-Early algorithm, and in particular establish that it uses a linear
number of cuts. In \Cref{subsec:incentivization}, we briefly
review the standard way to handle players not answering queries honestly. In \Cref{subsec:zero_value}, we briefly discuss how to handle players with zero-value regions.

\subsection{Strengthening the Inductive Hypothesis and the Subsistence Condition}
\label{subsec:generalizedinduction}

We strengthen the inductive hypothesis so that the Late-Early algorithm solves the following problem, which we will call the generalized cake cutting problem.
Syntactically an instance $J$ of the generalized cake cutting problem that is given to the algorithm is
a tuple of the form $((c, d), k, n, U, T)$ consisting of:
\begin{itemize}
    \item A subinterval $(c, d)$ of the unit interval
    \item A positive integer $n$ 
    \item  A nonempty collection $T$ of players
    \item  A (possibly empty) collection $U$ of players that is disjoint from $T$
    \item A number of players $k = |T|+|U|$
\end{itemize}

A feasible solution is a proportionally fair assignment, 
which in this context is a partition of the interval $(c, d)$ into $k$ parts, and
an assignment of these parts to the $k$ players $T \cup U$, such that each player values their assigned
part at least $\frac{1}{n}$. 
In the initial call to the Late-Early algorithm,  $(c, d)$ is the whole cake $(0, 1)$, $n$ is the number of players, $k=n$, $T$ is the collection of all players,
and $U$ is the empty set. Note that the algorithm is still not told the underlying value functions $\mu_1, \ldots, \mu_k$ for
the $k$ players. 

Foreshadowing slightly, the algorithm that we design will further guarantee that the instances
of the generalized cake cutting problem satisfy an additional property, that we call the subsistence condition.
An instance of the generalized cake cutting problem satisfying the subsistence condition has the following three properties:
\begin{itemize}
    \item All the players in $T$ value the
portion $(c, d)$ at least $\frac{k}{n}$, the same as they would if the instance satisfied the abundance condition.
\item In contrast, the players in
$U$   need only value the
portion $(c, d)$ at least $\frac{k-1}{n}$, which is slightly less than required by the abundance condition.
\item Most importantly, there is a player $p \in T$ such that if player $p$ was assigned a proportional
piece of value $\frac{1}{n}$ from the left of the cake, then the remaining instance would satisfy
the abundance condition (and thus have a proportionally fair allocation).
\end{itemize}

The subsistence condition will be formally stronger than these three properties alone, but the properties are a good intuitive starting point for understanding subsistence. We now formally define the subsistence condition, and some notation that we will need in the design and analysis
of our algorithm. A pictorial illustration of these definitions is given in Figure \ref{fig:xyz}.

\begin{figure}[h!]
\begin{center}
\begin{tikzpicture}
\filldraw[black] (0,0) circle (2pt) node[anchor=north]{$c$};
\draw[black, thick] (0,0) -- (2,0);
\filldraw[black] (2,0) circle (2pt) node[anchor=north]{$x_p$};
\filldraw[black] (3,0) circle (2pt) node[anchor=north]{$s$};
\filldraw[black] (4,0) circle (2pt) node[anchor=north]{$y_{r,1}$};
\draw[black, thick] (4,0) -- (6,0);
\filldraw[black] (6,0) circle (2pt) node[anchor=north]{$y_{r,2}$};
\draw[black, thick] (6,0) -- (8,0);
\filldraw[black] (8,0) circle (2pt) node[anchor=north]{$y_{r,3}$};
\draw[dotted] (8,0) -- (10,0);
\filldraw[black] (10,0) circle (2pt) node[anchor=north]{$y_{r,k-1}$};
\draw[black, thick] (10,0) -- (12,0);
\filldraw[black] (12,0) circle (2pt) node[anchor=north]{$d$};
\end{tikzpicture}
\end{center}
\caption{An example of an instance of the generalized cake cutting problem that satisfies the subsistence condition. 
In particular, this figure illustrates the  relative ordering of the
positions $x_p$, $s$, and $y_{r,1}, y_{r,2}, \ldots,  y_{r,k-1}$ for players $p \in T$ and $r \in U$. Each horizontal line segment is a portion of cake that the player  values at $\frac{1}{n}$.
Note that $r$ may value $(c, y_{r,1})$ at less than $1 / n$, so it could be the case that $x_r > y_{r,1}$ (or not). 
For this reason we are not concerned with the position of $x_r$ for a player $r \in U$ and we do not show it in the figure. 
}
\label{fig:xyz}
\end{figure}
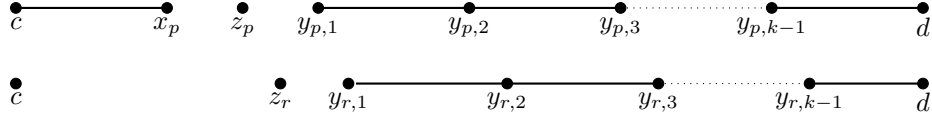

\begin{definition}\label{def:standard}
\hfill
\begin{itemize}
\item For each $p \in T \cup U$, define $x_p$ to be the point in $[c, d]$ where 
$V_p(c, x_p) = \frac{1}{n}$.
 \item For each player $p \in T \cup U$
 and each integer $j \in [1,k-1]$, define $y_{p,j}$ to be the point in $[c, d]$ where
 $V_p(y_{p,j}, d) = \frac{k-j}{n}$.
    \item For an integer $j \in [1, k-1]$, let $\prec_j$ be a total ordering of the $k$ players in
    $T \cup U$ such that 
    $p \prec_j q$ if and only if $y_{p,j} < y_{q,j}$, or     $y_{p,j} = y_{q,j}$ and $p < q$. 
    \item  A player $p$ is early in  $\prec_j$ if it is one of the $j$ smallest  points in $\prec_j$,
    that is if $| \{ q \mid q \prec_j p \}| \le j-1$.   
    \item  Otherwise a player $p$ is late in a $\prec_j$ if it is not one of the $j$ smallest  points in $\prec_j$, that is if $| \{ q \mid q \prec_j p \}| > j-1$.  
\end{itemize}
\end{definition}

\begin{definition}~
An  instance $J=((c,d), k, n, U, T)$ of the generalized cake cutting problem satisfies the subsistence condition if there exists a position $s \in [c, d]$ where all of the following hold:
\begin{itemize}
    \item For all $p \in T$ it is the case that $V_p(c, d) \geq \frac{k}{n}$.
    \item For all $p \in U$ it is the case that $V_p(s, d) \ge \frac{k-1}{n}$, or equivalently if $s \le y_{p,1}$.
     \item For all $p \in T $ it is the case that $V_p(c, s) \ge \frac{1}{n}$, or equivalently if $s \ge x_p$.
\end{itemize}
\end{definition}

Note that if the subsistence condition holds, then $x_p$ and $y_{p, j}$ are well-defined for each $p \in T \cup U$ and $j \in [1, k-1]$.

\subsection{The Algorithm Description}
\label{subsec:description}

At a high level, the Late-Early algorithm first picks two arbitrary players $p$ and $q$
from $T$, and then determines if either $p$ or $q$  values either the left or
right ends of the portion $(c, d)$ of cake more than any other player. 
If so, that player cuts a piece of value $\frac{1}{n}$ off the end that they highly value,
and the rest of the players recurse on the rest of the cake. 
If not, the algorithm uses doubling
binary search to find a $j$ such that both $p$ and $q$ are early in $\prec_{j+1}$
and at least one of $p$ or $q$ is late in $\prec_j$. The algorithm then cuts
the  cake $(c, d)$ at $\max\{y_{p,j}, y_{q,j}\}$, splits the players into
a collection of $j$ players who will share the portion 
of the cake to the left of this cut and a collection of $k-j$ players 
who will share the portion to the right of this cut, and then recurses
on these two subproblems.

\begin{algorithm}[H]
\caption{Late-Early }\label{alg:cap}
\begin{algorithmic}[1]
\Require{ An instance $J=((c,d), k, n, U, T)$ of the generalized cake cutting problem satisfying the subsistence condition}
\Ensure{A proportionally fair assignment}
\State Let $p$ be an arbitrary player in $T$
\State If $|T| > 1$ then let $q$ be an arbitrary player in $T$ other than $p$
\State \Comment{First check the base case $k=1$ and whether $p$ and $q$ can safely cut off a piece from the left or right}
\If{ $k =1$}  \Comment{$k = 1$ implies that $T = \{p\}$ and $U = \emptyset$}
\State Assign $p$ the subinterval $(c, d)$ \label{line:base}
\ElsIf{$|T|=1$ or $p $  is early in $\prec_1$} \label{line:T1}
\State Assign $p$ the subinterval $(c, x_p)$
\State Call Late-Early on the instance $J_1 = ((x_p, d), k-1, n,\emptyset, T \cup U-\{p\})$ \label{line:J1}
\ElsIf{$p$ is late in $\prec_{k-1}$}    \Comment{$|T| \geq 2$ for this case and subsequent ones}
\State Assign $p$ the subinterval $(y_{p,k-1}, d)$
\State Call Late-Early on the instance $J_2 = ((c, y_{p,k-1}), k-1, n, U, T-\{p\})$ \label{line:J2}
\ElsIf{$q$ is early in $\prec_1$ }
\State Assign $q$ the subinterval $(c, x_q)$
\State Call Late-Early on the instance $J_3=((x_q, d), k-1, n, \emptyset, T \cup U -\{q\})$ \label{line:J3}
\ElsIf{$q$ is late in $\prec_{k-1}$}
\State Assign $q$ the subinterval $(y_{q,k-1}, d)$
\State Call Late-Early on the instance $J_4=((c, y_{q,k-1}), k-1, n, U, T-\{q\})$ \label{line:J4}
\Else \Comment{Now prepare to recurse on two subproblems}
\State Using doubling search find a $j$ where at least one of $p$ and $q$ is late in $\prec_j$ and 
both $p$ and $q$ are early in $\prec_{j+1}$ \label{line:doubling}
\State If $p \prec_j q$ then rename $p$ to $q$ and vice versa\label{line:renaming}\Comment{We now know $q \prec_j p$}
\State Let $R := \{r \in T \cup U: p \prec_j r \text{ and } p \prec_{j+1} r \}$
\State Let $L := \{r \in T \cup U: r \prec_j p \text{ and } r \prec_{j+1} p \}$
\State Let $O := \{r \in T \cup U: r \prec_j p \text{ and } p \prec_{j+1} r \}$
\State Let $I := \{r \in T \cup U: p \prec_j r \text{ and } r \prec_{j+1} p \}$
\State Let $A$ be an arbitrary subcollection of $j- |L \cup \{q\}| $ players in $O-\{q\}$ \Comment{Note $q \in L \cup O$}
\State Let $B = O-\{q\}-A$
\State  Call Late-Early on the instance $J_5=( (c, y_{p,j}), j, n, U \cap (L \cup A ), T \cap (L \cup A \cup \{q\}))$ \label{line:J5}
\State  Call Late-Early on the instance $J_6= ( (y_{p,j}, d), k-j, n, B \cup R,  I \cup \{p\})$ \label{line:J6}
\EndIf
\end{algorithmic}
\end{algorithm}

\subsection{Implementation Details}
\label{subsec:implementation}

We now describe how to implement some of the steps in the Late-Early algorithm
using cuts and evaluations.
The position $x_p$  can be found 
by the operation $\cut_p(E_p(c) + \frac{1}{n})$. 
A $y_{p, j}$ can be found by the operation $\cut_p(E_p(d) - \frac{k-j}{n})$. 
The position of a player $p$ within a $\prec_j$ is one more than the number
of players $q \ne p$ where $E_q(d) - E_q(y_{p,j}) < \frac{k-j}{n}$ or
$E_q(d) - E_q(y_{p,j}) = \frac{k-j}{n}$ and $q < p$. 
Thus it takes one cut and $2k-2$ evaluations to determine
if a player $p$ is late or early in a $\prec_j$.

We now discuss how to implement the doubling search in line \ref{line:doubling} in the description
of the Late-Early algorithm. For simplicity we will assume $k=2^h$ is an integer power
of $2$. Let us define a $ \prec_j$ to be early if both $p$ and $q$ are
early in $\prec_j$, and define $\prec_j$ to be late otherwise. 
We know that right before line \ref{line:doubling} it is the case that $\prec_1$ is late
and $\prec_{k-1}$ is early. Thus we know that there must exist a $j$ such that $\prec_j$ is
late and $\prec_{j+1}$ is early. We can find such a $j$ using a standard doubling search. 
The algorithm iterates over the integers $i = 1, 2, 4, .., 2^{h-1}$.
On iteration $i$ the algorithm tests whether $\prec_i$ and $\prec_{k-i}$ are late or early. 
If $\prec_i$ is found to be early then the algorithm sets $a=i/2$ and $b=i$ and halts iterating over $i$,
else if $\prec_{k-i}$ is found to be late 
the algorithm sets $a=k-i$ and $b=k-i/2$ and halts iterating over $i$.
Note that as $\prec_{2^{h-1}}$ must be either late or early, $a$ and $b$ will be defined.
The algorithm then does a binary search over the interval $[a, b]$, going left if $\prec_m$ is early at the midpoint $m$, and going right otherwise. 
So to find $j$ requires $O(\log \min(j, k-j))$ cuts and $O(k \log \min(j, k-j))$ evaluations. Note that if $k$ is not an integer power of $2$, we iterate over $i$ until $a, b$ are defined or $i \geq k-i$, at which point we just test $\prec_{\lfloor k/2\rfloor}$. If $\prec_{\lfloor k/2\rfloor}$ is early then we set $a = i/2$ and $b = \lfloor k/2 \rfloor$, and if $\prec_{\lfloor k/2 \rfloor}$ is late then we set $a = \lfloor k/2\rfloor$ and $b = k-i/2$.

Lastly, we remark on the possible renaming of players $p$ and $q$ in line \ref{line:renaming}. If $p$ is renamed to $q$ and vice versa, the original identities of $p$ and $q$ should be maintained for the tie-breaking rules in $\prec_j$ and $\prec_{j+1}$, as we want $p$ and $q$'s exact positions in $\prec_j$ and $\prec_{j+1}$ to swap. An alternate clunky fix would be to break the $p \prec_j q$ and $q \prec_j p$ cases into two cases, and then have essentially identical code for those two cases.

\subsection{Correctness Analysis}
\label{subsec:correctness} 

To establish correctness of the Late-Early algorithm it will be sufficient to observe that the algorithm
is correct for instances with the subsistence condition that fall into the base case ($k=1$), to show by induction that each of the recursive
calls are to instances that satisfy  the subsistence condition, and to show that the original instance of cake cutting satisfies the subsistence condition.
We then show that if instance $J$ has the subsistence condition then each of the  instances $J_1, J_2, J_3, J_4, J_5, J_6$  have the 
subsistence condition if the Late-Early algorithm recurses on that instance. 
But first we  make some useful observations about the relationship
between the abundance condition and the subsistence condition.

\begin{definition}~
An  instance $J=((c,d), k, n, \emptyset, T)$ satisfies the abundance condition if for each player  $p \in T$ it is the case that $V_p(c, d) \ge \frac{k}{n}$. 
\end{definition}

\begin{observation}
\label{obs:subab}
    If an instance $J=((c,d), k, n, \emptyset, T)$ satisfies the abundance condition then it
    satisfies the subsistence condition.
\end{observation}
\begin{proof}
   It is sufficient to set $s = d$. The subsistence condition is clearly satisfied.
\end{proof}

\begin{observation} \label{obs:skinnyleft}
Consider an instance $J=((c,d), k, n, U, T)$ that satisfies the subsistence condition where $k \geq 2$. Assume there exists a $p \in T$ such that $p$ is early in $\prec_1$. 
Then the instance $((x_p,d), k-1, n, \emptyset, T \cup U - \{p\})$
satisfies the abundance (and hence the subsistence) condition. 
\end{observation}
\begin{proof}
This follows from the fact that $x_p \le y_{p,1}$, since $p \in T$, and from the fact that $p = \argmin_{r \in T \cup U} y_{r,1}$. 
\end{proof}

\begin{observation}
    The initial instance, in which $(c, d)$ is the whole cake $(0, 1)$, $n$ is the number of players, $k=n$, $T$ is the collection of all the players,
and $U$ is the empty set, satisfies the abundance (and hence the subsistence) condition.  
\end{observation}

\begin{observation}
    The algorithm Late-Early is correct on instances that satisfy the subsistence condition and have  $k=1$. 
\end{observation}
\begin{proof}
The lone player is in $T$, which values $(c, d)$ at least $\frac{1}{n}$ by the subsistence condition.
\end{proof}

\begin{lemma} If the algorithm Late-Early makes a recursive call on the instance
 $J_1$ in line \ref{line:J1}, then $J_1 = ((x_p, d), k-1, n,\emptyset, T \cup U-\{p\})$ is  an instance that satisfies  the abundance (and hence the subsistence) condition. 
\end{lemma}
\begin{proof}
    First consider the case that $p$ is early  in $\prec_1$.  Then the claim follows by Observation \ref{obs:skinnyleft}. 
  Next consider the case that $|T| = 1$, so  $T = \{p\}$. Note that it must be the case that $U \neq \emptyset$, since otherwise line \ref{line:T1} of the algorithm is not reached. Therefore, $ T \cup U - \{p\} = U \neq \emptyset$, so $J_1$ is a 
  syntactically valid instance.  Now note that because $J$ satisfies the subsistence condition it must be the case
  that $x_p \le \min_{r \in U} y_{r,1}$. Thus each player in $U$ must value the portion $(x_p, d)$ at least $\frac{k-1}{n}$
  by the subsistence condition. Thus $J_1$ satisfies the abundance condition, and thus by Observation \ref{obs:subab} the subsistence condition.
\end{proof}

\begin{lemma} \label{lem:pskinnyright} If the algorithm Late-Early makes a recursive call on the instance
 $J_2 = ((c, y_{p,k-1}), k-1, n, U, T-\{p\})$ in line \ref{line:J2}, then $J_2$ is an instance that satisfies the subsistence condition. 
\end{lemma}
\begin{proof}
  First observe that the set $T$ in $J$ has cardinality at least two if a recursive
  call was made on $J_2$. Thus  $T-\{p\}$ is not empty, and thus $J_2$ is a syntactically correct instance. 
  We now 
   show that $J_2$ satisfies the subsistence condition. To accomplish this it is sufficient to set $\hat{s} = \min\{s, y_{p, k-1}\}$ where $\hat{s}$ is the separator for $J_2$ and $s$ is the separator for $J$. The first criterion of the subsistence condition for $J_2$ holds because $p = \argmax_{r \in T \cup U} y_{r, k-1}$; therefore, 
  $V_r(c, y_{p, k-1}) \geq V_r(c, y_{r, k-1}) = V_r(c, d) - V_r(y_{r, k-1}, d) \geq \frac{k-1}{n}$ for all $r \in T-\{p\}$. For the second criterion of the subsistence condition for $J_2$, note that if $U = \emptyset$ then the second criterion vacuously holds, and if $U \neq \emptyset$ then $\hat{s} = s$. Otherwise, we have $s > \hat{s} = y_{p, k-1} \geq y_{r, k-1} \geq y_{r, 1}$ for all $r \in U$, contradicting the necessity of $s \leq y_{r, 1}$ for the subsistence condition for $J$. Thus if $U \neq \emptyset$, we have $V_r(s, y_{p, k-1}) \geq V_r(s, y_{r, k-1}) = V_r(s, d) - V_r(y_{r, k-1}, d) \geq \frac{k-2}{n}$ for all $r \in U$. The third criterion of the subsistence condition for $J_2$ holds if $\hat{s} = s$ because $s$ satisfies the third criterion of the subsistence condition for $J$, and no cake was cut from the left. Further, it trivially holds if $\hat{s} = y_{p, k-1}$, because this is the right endpoint of the interval for $J_2$, and the first criterion of the subsistence condition already guarantees that players $r \in T - \{p\}$ value the interval at least $\frac{1}{n}$.
\end{proof}

\begin{lemma} If the algorithm Late-Early makes a recursive call on the instance
$J_3=((x_q, d), k-1, n, \emptyset, T \cup U -\{q\})$ in line \ref{line:J3}, then $J_3$ is an instance that satisfies  the abundance (and hence the subsistence) condition. 
\end{lemma}
\begin{proof} Note that if a recursive call to $J_3$ is made it must be
the case that $|T| >1$, and thus $q$ exists. 
   Then this is an immediate consequence of Observation \ref{obs:skinnyleft}.
\end{proof}

\begin{lemma} If the algorithm Late-Early makes a recursive call on the instance
 $J_4$ in line \ref{line:J4}, then $J_4=((c, y_{q,k-1}), k-1, n, U, T-\{q\})$ is an instance that satisfies the subsistence condition. 
\end{lemma}
\begin{proof} The proof is the same as the proof of Lemma \ref{lem:pskinnyright}.
\end{proof}

\begin{figure}[h]
\begin{center}
\begin{tikzpicture}
\filldraw[black] (-2,0) circle (2pt) node[anchor=north]{$y_{q,j}$};
\filldraw[black] (0,0) circle (2pt) node[anchor=north]{$y_{p,j}$};
\draw[black, thick] (0,0) -- (6,0);
\filldraw[black] (6,0) circle (2pt) node[anchor=north]{$y_{p,j+1}$};
\filldraw[black] (-2,-1) circle (2pt) node[anchor=east]{$L$};
\draw[black, thick] (-2,-1) -- (4,-1);
\filldraw[black] (4,-1) circle (2pt) ;
\filldraw[black] (2,-2) circle (2pt) node[anchor=east]{$R$};
\draw[black, thick] (2,-2) -- (8,-2);
\filldraw[black] (8,-2) circle (2pt) ;
\filldraw[black] (2,-3) circle (2pt) node[anchor=east]{$I$};
\draw[black, thick] (2,-3) -- (4,-3);
\filldraw[black] (4,-3) circle (2pt) ;
\filldraw[black] (-2,-4) circle (2pt) node[anchor=east]{$O$};
\draw[black, thick] (-2,-4) -- (8,-4);
\filldraw[black] (8,-4) circle (2pt) ;
\end{tikzpicture}
\end{center}
\caption{The relative positions of $y_{q,j}$, $y_{p,j}$ and $y_{p,j+1}$ and the positions of intervals of the form $(y_{r, j}, y_{r,j+1})$ for $r$ in $L$, $R$, $I$ and $O$, relative to
the interval $(y_{p, j}, y_{p,j+1})$. Note the ends of some pair from $L$, $R$, $I$ and $O$ being vertically aligned signifies that no particular ordering of these ends is implied by the definition.}
\label{fig:simplelateearly}
\end{figure}
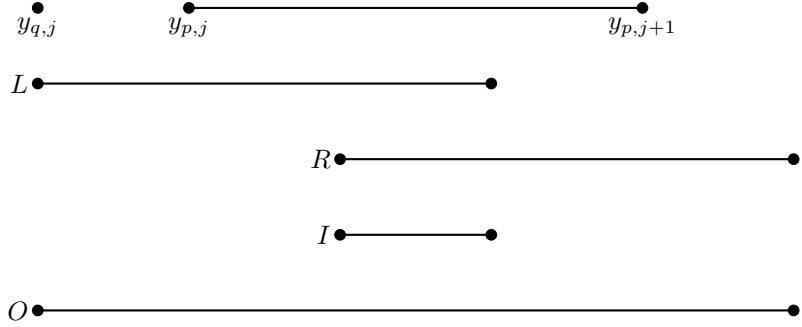

\begin{lemma} \label{lem:J5} If the algorithm Late-Early makes a recursive call on the instance
 $J_5=( (c, y_{p,j}), j, n, U \cap (L \cup A ), T \cap (L \cup A \cup \{q\}))$ in line \ref{line:J5}, then $J_5$ satisfies  the subsistence condition. 
\end{lemma}
\begin{proof} 
To aid in our exposition we will use $\hat T$ and $\hat U$ for the values of $T$ and $U$ in instance $J_5$, that is $\hat U = U \cap (L \cup A )$ and $\hat T = T \cap (L \cup A \cup \{q\})$. The reader should also reference \Cref{fig:simplelateearly} as a helpful visual aid. We first need to establish that $J_5$ is a syntactically correct input by demonstrating that  there are $j$ players,
and that $\hat T$ is not empty. We then need to establish that
$J_5$ satisfies the subsistence condition. 

The set  $\hat T$ is not empty because it contains $q$. 
Let us now turn to counting the number of players in the instance $J_5$.
    First note that the cardinality of $L \cup I$ is at most $j$. This is because  $p$ is early in $\prec_{j+1}$ and every element in $L \cup I$ precedes $p$ in $\prec_{j+1}$.
Then note that the cardinality of $L \cup \{q\}$ is at most $j$.
To see this first note that if $|L| < j$, then clearly the cardinality of $L \cup \{q\}$ is at most $j$. Else consider the case that $|L| = j$.
Then it  must be that $I$ is empty. Thus $y_{p, j+1}$ must be exactly the $(j+1)$'st smallest 
position in $\prec_{j+1}$, and the $j$ positions before $y_{p, j+1}$ in $\prec_{j+1}$ are the positions $y_{r, j+1}$ associated with the $r \in L$. But then as $q$ is early in $\prec_{j+1}$ 
it must be the case that $q \in L$. And thus we can conclude that the cardinality of $L \cup \{q\}$ is at most $j$.
Since $p$ is late in $\prec_j$, the cardinality of $O$ must be at least $j - |L|$. Thus 
there are enough players in $O-\{q\}$ to form the set $A$ of $j - | L \cup \{q\}|$ players. 
Thus we can conclude that $\hat T \cup \hat U$ contains $j$ players in the two cases $q \in L$ or $q \in O$.

Now we need to establish that instance $J_5$ satisfies the subsistence condition. 
To accomplish this it is sufficient to set $\hat{s} = \min\{s, y_{p, j}\}$, where $\hat{s}$ is the separator in $J_5$ and $s$ is the separator in $J$. In this setting the first criterion of the subsistence condition is that each player $r \in \hat{T}$ values $(c, y_{p, j})$ at least $\frac{j}{n}$. Note that each player $r \in \hat{T}$ is in $L \cup A \cup \{q\} \subseteq L \cup O$. Thus

\[ V_r(c, y_{p, j}) = V_r(c, d) - V_r(y_{p, j}, d) \geq V_r(c, d) - V_r(y_{r, j}, d) \geq \frac{k}{n} - \frac{k-j}{n} = \frac{j}{n} \]

The first inequality follows by the fact that $r \in L \cup O$, so $r \prec_j p$ and thus $y_{r, j} \leq y_{p, j}$. We then use the first condition of the subsistence condition on $J$ and the definition of $y_{r, j}$. The second criterion of the subsistence condition is that each player $r \in \hat{U}$ values $(\hat{s}, y_{p, j})$ at least $\frac{j-1}{n}$. If $\hat{U} = \emptyset$ then the second criterion is vacuously true, and if $\hat{U} \neq \emptyset$ then we must have $\hat{s} = s$, as otherwise $s > \hat{s} = y_{p, j} \geq y_{r, j} \geq y_{r, 1}$ for all $r \in \hat{U}$, contradicting the necessity of $s \leq y_{r, 1}$ for the subsistence condition in $J$. Thus for all $r \in \hat{U}$,

\[
\begin{aligned}
V_r(\hat{s}, y_{p, j})
&= V_r(s, y_{p, j}) \geq V_r(s, y_{r, j}) \geq V_r(y_{r, 1}, y_{r, j}) \\
&= V_r(y_{r, 1}, d) - V_r(y_{r, j}, d) \geq \frac{k-1}{n} - \frac{k-j}{n} = \frac{j-1}{n}
\end{aligned}
\]

The first inequality follows by $y_{r, j} \leq y_{p, j}$. The second inequality follows by $s \leq y_{r, 1}$. We then use the definitions of $y_{r, 1}$ and $y_{r, j}$. The third criterion of the subsistence condition for $J_5$ holds if $\hat{s} = s$ because $s$ satisfies the third criterion of the subsistence condition for $J$, and no cake was cut from the left. Further, it trivially holds if $\hat{s} = y_{p, j}$, because this is the right endpoint of the interval for $J_5$, and the first criterion of the subsistence condition already guarantees that players $r \in \hat{T}$ value the interval at least $\frac{1}{n}$.
\end{proof}

\begin{lemma} If the algorithm Late-Early makes a recursive call on the instance
 $J_6= ( (y_{p,j}, d), k-j, n, B \cup R,  I \cup \{p\})$  in line \ref{line:J6}, then $J_6$ is  an instance that satisfies   the subsistence condition. 
\end{lemma}
\begin{proof} 
To aid in our exposition we will use $\hat T$ and $\hat U$ for
the values of $T$ and $U$ in instance $J_6$, that is $\hat U = B \cup R$ and $\hat T =  I \cup \{p\}$. The reader should also reference \Cref{fig:simplelateearly} as a helpful visual aid.
We first need to establish that $J_6$ is a syntactically correct input by demonstrating that  there are $k-j$ players,
and $\hat T$ is not empty. We then need to establish that
$J_6$ satisfies the subsistence condition.

The set $\hat T$ is not empty because it contains $p$. 
The number of players  in $J_6$ must be $k-j$ as instance $J$ contained $k$ players,
it was established in the proof of Lemma \ref{lem:J5} that the instance $J_5$ contains $j$
players, and all players in instance $J$ were assigned to either instance $J_5$ or $J_6$. 
To see this note that player $q$ and all players in $A \cup L$ were assigned to instance $J_5$,
the players in instance $J_6$ are player $p$ and the players in $B=O-A-\{q\}$, $I$ and $R$,
and $I \cup O \cup L \cup R$ covers all players in $J$ besides $p$.

Now we need to establish that instance $J_6$ satisfies the subsistence condition. 
To accomplish this we set $s = y_{p,j+1}$. In this setting the first criterion of the subsistence condition is that each player $r \in \hat{T}$ values $(y_{p, j}, d)$ at least $\frac{k-j}{n}$. Note that $\hat{T} = I \cup \{p\}$. Thus for all $r \in \hat{T}$,

\[ V_r(y_{p, j}, d) \geq V_r(y_{r, j}, d) \geq \frac{k-j}{n} \]

The first inequality follows by the fact that $r \in I \cup \{p\}$. If $r = p$ then the statement is trivial, and if $r \in I$ then $r \succ_j p$, so $y_{r, j} \geq y_{p, j}$. The second inequality follows by the definition of $y_{r, j}$. The second criterion of the subsistence condition is that each player $r \in \hat{U}$ values $(y_{p, j+1}, d)$ at least $\frac{k-j-1}{n}$. Note that each player $r \in \hat{U}$ is in $B \cup R \subseteq O \cup R$. Thus

\[ V_r(y_{p, j+1}, d) \geq V_r(y_{r, j+1}, d) = \frac{k-(j+1)}{n} = \frac{k-j-1}{n} \]

The first inequality follows by the fact that $r \in O \cup R$, so $r \succ_{j+1} p$ and thus $y_{r, j+1} \geq y_{p, j+1}$. The third criterion of the subsistence condition is that each player $r \in \hat{T}$ values $(y_{p, j}, y_{p, j+1})$ at least $\frac{1}{n}$. Note that $\hat{T} = I \cup \{p\}$. Thus

\[ V_r(y_{p, j}, y_{p, j+1}) \geq V_r(y_{r, j}, y_{r, j+1}) = \frac{1}{n} \]

The first inequality follows by the fact that $r \in I \cup \{p\}$. If $r = p$ then the statement is trivial, and if $r \in I$ then $r \succ_j p$ and $r \prec_{j+1} p$, so $y_{p, j} \leq y_{r, j} \leq y_{r, j+1} \leq y_{p, j+1}$.
\end{proof}

\subsection{Complexity Analysis}
\label{subsec:complexity}

As finding a $j$ where $\prec_j$ is late and $\prec_{j+1}$ is early takes
$O(\log \min(j, n-j))$ cuts, we
can bound the number of cuts made by our Late-Early algorithm  by the recurrence relation:
$$C(n) \le \max_{j \in [1, n-1] } ( C(j) + C(n-j) + O(\log \min(j, n-j)))$$
One can then prove by induction that 
$C(n) \le a n - a \log n - a$ for some sufficiently large constant $a$.
As finding a $j$ where $\prec_j$ is late and $\prec_{j+1}$ is early takes
$O(n\log \min(j, n-j))$  evaluations, we
can bound the number of evaluations made by our Late-Early algorithm by the recurrence relation:
$$E(n) \le \max_{j \in [1, n-1] } ( E(j) + E(n-j) + O(n \log  \min(j, n-j)))$$
Similarly, one can then prove by induction that $E(n) \leq an^2$ for some sufficiently large constant $a$. Thus $C(n) = O(n)$ and $E(n) = O(n^2)$.

\subsection{Handling Dishonest Players}
\label{subsec:incentivization}

Note that the most common  way  that the cake cutting problem is defined
does not require that
the players answer the queries honestly. But then  the algorithm is only required
to 
assign a proportionally fair piece to each player that always answered truthfully.
However, one can always convert an algorithm that assumes that players answer honestly
to one that assigns proportionally fair pieces to truthful players in the following standard way.
As long as each player $p$ answers in such a way that there is a $\mu_p$ that
is consistent with their answers, then the algorithm's assignment will be proportionally
fair to each player that answered honestly. If at some point a player $p$ answers
a query in such a way that there is no $\mu_p$ consistent with the answers up to this point,
then the algorithm can let $\hat \mu_p$ be an arbitrary value function that is consistent
with all of $p$'s answers previously, and then the algorithm can use $\hat \mu_p$ for this query 
and all subsequent queries to player $p$. Again this will be  proportionally
fair to each player that answered honestly.

\subsection{Handling Zero-Value Regions}\label{subsec:zero_value}
If players $p$ have zero value regions $V_p(a, b) = 0$, then for $\alpha = E_p(a) = E_p(b)$, the query Cut$_p(\alpha)$ could return an arbitrary point in $(a, b)$. A call to the same Cut query could even return different points in $(a, b)$. Further, there are a few formal complications that arise in our standard definitions (\Cref{def:standard}). For example, if there is a zero-value interval $(a, b)$ where $V_p(a, d) = V_p(b, d) = \frac{k-j}{n}$, then $y_{p, j}$ could be defined as any point in $(a, b)$. This means that the total orderings $\prec_j$ are not well-defined. Further, this cannot be solved by simply fixing $y_{p, j}$ as the leftmost, rightmost, or arbitrary point with the correct cake value, because as mentioned the actual call to Cut$_p(E_p(d) - \frac{k-j}{n})$ will return an arbitrary point in $(a, b)$.

The solution is that the values $x_p$ and $y_{p, j}$ are chosen on demand and cached in each recursive call. Specifically,

\begin{itemize}
    \item If a player $p$ has $V_p(y_{q, j}, d) = E_p(d) - E_p(y_{q, j}) = \frac{k-j}{n}$ where $y_{q, j}$ is already cached but $y_{p, j}$ is not, then we cache $y_{p, j} = y_{q, j}$, and do not perform a Cut query on $p$ to find $y_{p, j}$ later.
    \item When $x_p$ is needed for a player $p \in T$, we compute $z_p = \text{Cut}_p\left(E_p(c) + \frac{1}{n}\right)$, find $y_{p, 1}$ (by referencing the cached value if it exists, or performing a Cut query), and set $x_p = \min\{z_p, y_{p, 1}, s\}$.
    \item We pass the separator $s$ and the cached values explicitly down the recursive calls.
\end{itemize}

Now, when comparing players in $\prec_j$, the already-cached values $y_{p, j}$ should be used if available. Note that the algorithm can still determine whether a player $p$ is late or early in $\prec_j$ by referencing the cached value $y_{p, j}$ if it exists, or by performing the same single query Cut$_p(E_p(d) - \frac{k-j}{n})$. These fixes keep the precedence constraints, separators, and subsistence condition consistent over time when there are zero-value regions.

\section{ Other Related Work}
\label{sec:nonessential}

Because the cake cutting literature is vast, 
we cannot hope to summarize all of it. Instead, we will just give a small sample that hopefully gives the flavor of the vast literature. 
Reasonable jumping off points for further investigation are~\cite{RobertsonWebbbook,Procaccia_2016,Brams_Taylor_1996,enwiki:1352137217}.

There is literature on approximate  proportional fairness, which in the context of proportional
fairness would mean that each player need only be given a piece that they
value at least $\frac{1}{cn}$ for some $c > 1$. 
For example, 
 \cite{JiriGerhardApprox} gives a deterministic algorithm with a linear number of cuts that achieves
constant approximate fairness, and \cite{EdmondsP06} gives a randomized algorithm with a linear number of queries that achieves
constant approximate fairness.

The other major criterion for fairness in the literature is envy-freeness,
which means that each player values their piece at least as much as they
value the piece that any other player received. 
There is an $\Omega(n^2)$ lower bound on the number of queries needed
for a deterministic algorithm to achieve envy-freeness~\cite{procaccia}. The first bounded deterministic algorithm to achieve envy-freeness for arbitrary $n$ used a super-exponential number of queries~\cite{aziz-mackenzie}. Recently, a deterministic algorithm has been found which uses an exponential number of queries~\cite{YeBai2026}.

There is an algorithm that uses $O(n^3)$ cuts in the setting that a proportionally fair division
has to be maintained as the players arrive online over time, and previously assigned cake can be reassigned~\cite{fink1964note}.

There are algorithms that guarantee strong proportional fairness, that is each
player gets a piece they value strictly more than $\frac{1}{n}$, on instances
where some two players disagree about the value of some portion of the cake~\cite{woodall1986note,enwiki:1324566255}.

There are continuous, moving-knife, algorithms that do not fit into the 
Robertson-Webb model~\cite{enwiki:1326284156}. For example, 
Dubins and Spanier~\cite{DubinsSpanier1961} give a moving-knife implementation of
the Last Diminisher algorithm, in which each player is expected to signal when
their value reaches $\frac{1}{n}$ as the algorithm continuously moves a knife
from one end of the cake to the other. 

There is a large literature on fairly dividing discrete items, see for example \cite{enwiki:1346798097}.

\section{Conclusion}

The title of this paper, ``No, Cake Cutting Really is a Piece of Cake'' is in direct response to the paper of Edmonds and Pruhs~\cite{EdmondsPruhs}, titled, ``Cake Cutting Really is Not a Piece of Cake'', 
which
showed an $\Omega(n \log n)$ lower bound on the number of cut queries plus the number of evaluation
queries for any deterministic algorithm. So under the interpretation of ``piece of cake'' as $O(n)$  queries,
one intuitive takeaway from these two papers is whether or not cake cutting is ``a piece of cake'' for
deterministic algorithms
depends on whether or not one counts evaluation queries. 

We believe the two most important/interesting remaining open algorithmic questions related to proportionally fair cake cutting are the following.

\paragraph{Open Question 1:} What is the optimal expected number of cut queries plus
evaluation queries to achieve proportional fairness for a randomized algorithm? 
In particular, is there a randomized algorithm that uses $O(n)$ cut queries
plus evaluation queries in expectation? Or is there an $\Omega(n \log n)$ lower bound 
on the expected number of cut queries plus evaluation queries for a randomized algorithm?
The $\Omega(n \log n)$ lower bound for deterministic algorithms that
assign subintervals as pieces can be extended to randomized algorithms~\cite{EdmondsPruhs,arndt2026omeganlognrandomized}. But the extension of
the deterministic  lower bound for general disconnected pieces  in \cite{EdmondsPruhs} does not seem to readily extend to   randomized algorithms. 

\paragraph{Open Question 2:} For deterministic algorithms, what is the optimal tradeoff between the number of
cut queries and the number of evaluation queries, when the number of cut queries
is in the range from $\Omega(n)$ to $O(n \log n)$? In particular,
when the number of cut queries is $O(n)$, can the number of evaluation
queries be $O(n \log n)$? If so, then this would arguably close the
story on deterministic algorithms for proportional fairness. Or alternatively, when the number of cut queries is $O(n)$, must the number of evaluation
queries be $\Omega(n^2)$?

Subsequent to the arXiv posting of this work (June 2026), Nisarg Shah announced a deterministic cake cutting algorithm that achieves proportional fairness using $O_{\epsilon}(n)$ cuts and $O_{\epsilon}(n^{1+\epsilon})$ evaluation queries for every fixed $\epsilon > 0$ with the assistance of GPT-5.6-Sol (August 2026), along with a batch of other AI-assisted fair division results, and plans to make the manuscripts more human-readable before release~\cite{Shah2026Announcement}. At the time of preparing this version (October 2026), we were not able to locate a publicly available proof.


\paragraph{Acknowledgments:} The last author is grateful to Jeff Edmonds, Jiri Sgall,
and the late, great Gerhard Woeginger, for many stimulating conversations about cake cutting, fair division, and  fun with algorithms.

\bibliographystyle{siamplain}
\bibliography{cakereferences}

\end{document}